\documentclass[final,twocolumn,10pt]{asme2ej}

\usepackage{setspace}
\usepackage{amsmath} 
\usepackage{amssymb,nicefrac}
\usepackage[dvips]{epsfig}
\usepackage[dvips]{graphicx}
\usepackage[section]{placeins}
\usepackage{fancyhdr}
\usepackage{color}
\usepackage{pstricks}
\usepackage{pst-node}
\usepackage{dsfont}
\usepackage{mathrsfs}
\usepackage{sgame}
\usepackage{float}
\usepackage[T1]{fontenc}
\usepackage{aurical}
\usepackage{tikz,pgfplots,verbatim}
\usetikzlibrary{spy,arrows,backgrounds,plotmarks}
\usepackage{subcaption}
\allowdisplaybreaks

\newif\iffigures
\figurestrue 

\iffigures
\usepgfplotslibrary{external} 
\newtheorem{proposition}{Proposition}[section]

\newtheorem{property}{Property}[section]

\newcommand{\tr}{^{\mathrm T}}

\newcommand{\df}{\doteq}

\newcommand{\magn}[1]{\left\vert #1 \right\vert}

\newcommand{\QQ}{\mathsf{Q}}
\newcommand{\PP}{\mathsf{P}}
\newcommand{\qq}{\mathsf{q}}

\title{Regression Models for Output Prediction of Thermal Dynamics in Buildings\thanks{This work is part of the mpcEnergy project which is supported within the program Regionale Wettbewerbsf\"{a}higkeit O\"{O} 2007-2013 by the European Fund for Regional Development as well as the State of Upper Austria. The research reported in this article has been (partly) supported by the Austrian Ministry for Transport, Innovation and Technology, the Federal Ministry of Science, Research and Economy, and the Province of Upper Austria in the frame of the COMET center SCCH. An earlier version of part of this paper appeared in \cite{ChasparisNatschlaeger14}.}}

\author{Georgios C. Chasparis\thanks{Address all correspondence related to this paper to this author.}
    \affiliation{
	Software Competence Center Hagenberg GmbH\\
	Department of Data Analysis Systems\\
	Softwarepark 21\\
	4232 Hagenberg, Austria\\
    Email: georgios.chasparis@scch.at
    }	
}

\author{Thomas Natschlaeger
    \affiliation{
	Software Competence Center Hagenberg GmbH\\
	Department of Data Analysis Systems\\
	Softwarepark 21\\
	4232 Hagenberg, Austria\\
    Email: thomas.natschlaeger@scch.at
    }	
}

\begin{document}

\maketitle

\begin{abstract}
Standard (black-box) regression models may not necessarily suffice for accurate identification and prediction of thermal dynamics in buildings. This is particularly apparent when either the flow rate or the inlet temperature of the thermal medium varies significantly with time. To this end, this paper analytically derives, using physical insight, and investigates linear regression models with nonlinear regressors for system identification and prediction of thermal dynamics in buildings. Comparison is performed with standard linear regression models with respect to both a) identification error, and b) prediction performance within a model-predictive-control implementation for climate control in a residential building. The implementation is performed through the EnergyPlus building simulator and demonstrates that a careful consideration of the nonlinear effects may provide significant benefits with respect to the power consumption.
\end{abstract}


\section{Introduction}

The increased demand for electricity power and/or fuel consumption in residential buildings requires an efficient control design for all heating/cooling equipment. To this end, recently, there have been several efforts towards a more efficient climate control in residential buildings \cite{NghiemPappas11,Oldewurtel12,NghiemPappas13,TouretzkyBaldea13}. Efficiency of power/fuel consumption is closely related to the ability of the heating/cooling strategy to incorporate predictions of the heat-mass transfer dynamics and exogenous thermal inputs (e.g., outdoor temperature, solar radiation, user activities, etc.). Naturally this observation leads to model-predictive control (MPC) implementations (see, e.g., \cite{NghiemPappas11,Oldewurtel12}). The performance though of any such implementation will be closely related to the performance of the prediction models for the evolution of both the indoor temperature and the exogenous thermal inputs.

A common approach in the derivation of prediction models for the indoor temperature evolution relies on the assumption that the heat-mass transfer dynamics can be approximated well by linear models. For example, in the optimal supervisory control formulation introduced in \cite{NghiemPappas11} for controlling residential buildings with an heating-ventilation-air-conditioning (HVAC) system, a linear (state-space) model is considered. As pointed out in \cite{NghiemPappas11}, including the actual dynamics might be complicated for the derivation of optimal control strategies. Similar is the assumption in the dynamic game formulation of HVAC thermally controlled buildings in \cite{Coogan13}, where again a linear model is assumed for the overall system. Also, in references \cite{Rasmussen05,Maasoumy14}, detailed nonlinear representations of the heat-mass transfer dynamics (of a vapor-compression system in \cite{Rasmussen05} and for a building in \cite{Maasoumy14}) are linearized about an operating point to allow for a more straightforward controller design. 

Given this difficulty in incorporating nonlinear representations in the controller design, several efforts for identification of heat-mass transfer dynamics have adopted linear transfer functions, such as the standard ARX or ARMAX black-box model structures (cf.,~\cite{Ljung99}). Examples of such implementations include the MIMO ARMAX model in \cite{YiuWang07}, the MISO ARMAX model for HVAC systems in \cite{Scotton13} and the ARX model structure considered in \cite{Malisani10}.

On the other hand, the nonlinear nature of the heat-mass transfer dynamics in buildings has been pointed out by several papers, including the switched linear dynamics for modeling the intermittent operation of RH systems in \cite{NghiemPappas13}, the multiple-time scale analysis considered in \cite{Malisani10,TouretzkyBaldea13} and the physics-based representation of the air-conditioning systems in \cite{Rasmussen05}. In reference \cite{Malisani10}, a comparison is performed between standard linear ARX models with two time-scale transfer models, which according to the authors better represent thermal models. Indeed, the operation of either an radiant-heating (RH) system or an HVAC system usually takes place in a faster time-scale than the room dynamics, an observation that was also utilized for the derivation of a (nonlinear) model predictive controller in \cite{TouretzkyBaldea13} using singular-perturbation techniques.


The following questions naturally emerge: \emph{What is the effect of the nonlinearities of an RH or an HVAC system in the identification error of the indoor temperature?} \emph{Can more detailed model representations be utilized to reduce identification error and improve energy efficiency?} This paper begins with the observation that thermal dynamics in buildings are nonlinear in nature due to the operation of the RH and HVAC systems. We avoid introducing any multiple time-scale approximation as in \cite{TouretzkyBaldea13} and we adopt a detailed representation of the heat-mass transfer dynamics (using standard Bond-graph models). In principle, linear transfer models incorporating linear regressors, such as Output-Error (OE), ARX or ARMAX structures, may be sufficient for system identification in buildings. However, we wish to investigate whether this is a reasonable assumption and whether more accurate representations of the dynamics are necessary.

To this end, linear regression models with nonlinear regressors are derived analytically using physical insight for prediction of the indoor temperature in residential buildings with RH and/or HVAC systems. The derivation accounts for alternative information structures depending on the sensors available, which is the main contribution of this paper. The derived identification models are compared against standard (black-box) linear regression models trained with simulated data generated in EnergyPlus (V7-2-0) building simulator tool developed by the U.S. Department of Energy \cite{EnergyPlus}. Comparison is also performed with respect to the prediction performance of the derived models when used within a standard MPC for climate control in a residential building. The results indicate that a more detailed design of the thermal-mass transfer prediction models can be beneficial with respect to the energy consumption. This paper extends prior work of the same authors \cite{ChasparisNatschlaeger14}, since it extends the derivation of the regression models into a larger family of information structures (sensors available), while it provides a detailed comparison with respect to the energy consumption.

In the remainder of the paper, Section~\ref{sec:Setup} provides a description of the overall framework and the objective of this paper. Section~\ref{sec:LinearRegressionModels} provides a short introduction to a class of linear regression models (briefly LRM) for system identification of heat-mass transfer dynamics. Section~\ref{sec:NonLinearRegressionModels} analytically derives using physical insight linear regression models with nonlinear regressors (briefly NRM). In Section~\ref{sec:PerformanceEvaluation}, we compare the performance of the derived NRM with LRM with respect to both a) identification error, and b) prediction performance within a standard MPC for climate control. Finally, Section~\ref{sec:Conclusions} presents concluding remarks.

\emph{Notation:} 
\begin{itemize}
 \item[$\bullet$] ${\rm col}\{x_1,...,x_n\}$, for some real numbers $x_1,...,x_n\in\mathbb{R}$, denotes the \emph{column vector} $(x_1,...,x_n)$ in $\mathbb{R}^{n}$. Also, ${\rm row}\{x_1,...,x_n\}$ denotes the corresponding \emph{row vector}.
 \item[$\bullet$] ${\rm diag}\{x_1,...,x_n\}$ denotes the \emph{diagonal matrix} in $\mathbb{R}^{n\times{n}}$ with diagonal entries $x_1,...,x_n$.
 \item[$\bullet$] for any finite set $A$, $\magn{A}$ denotes the \emph{cardinality} of $A$.
 \item[$\bullet$] $\doteq$ denotes \emph{definition}.
\end{itemize}

\section{Framework \& Objective}		\label{sec:Setup}

When addressing system identification and control of heat-mass transfer dynamics in buildings, the assumed model of indoor temperature prediction should be a) accurate enough to capture the phenomena involved, and b) simple enough to accommodate a better control design. Given that prior literature has primarily focused on linear state-space or transfer models, we wish to address the following:
\begin{enumerate}
\item investigate the performance of standard (black-box) linear (regression) models (LRM) in the identification/prediction of heat-mass transfer dynamics in residential buildings;
\item derive analytically, using physical insight and under alternative information structures, an accurate representation of the dynamics using linear regression models with nonlinear regressors (NRM); 
\item provide a detailed comparison between the derived NRM with standard LRM and assess the potential energy saving.
\end{enumerate}

To this end, we consider residential buildings divided into a set of zones $\mathcal{I}$. Each zone is controlled independently in terms of heating using either a RH and/or an HVAC system. We will use the notation $i,j$ to indicate representative elements of the set of zones $\mathcal{I}$. The temperature of a zone $i$ is affected by the temperature of the neighboring zones and/or the outdoor environment, denoted $\mathcal{N}_i$, i.e., those zones which are adjacent to $i$, separated by some form of separator. 

In the remainder of this section, we provide some standard background on heat-mass transfer dynamics (based on which the derivation of the prediction models will be presented). We further discuss the considered assumptions with respect to the available sensors, and finally some generic properties of the dynamics.

\begin{figure*}[t!]
\centering
\includegraphics[scale=.45]{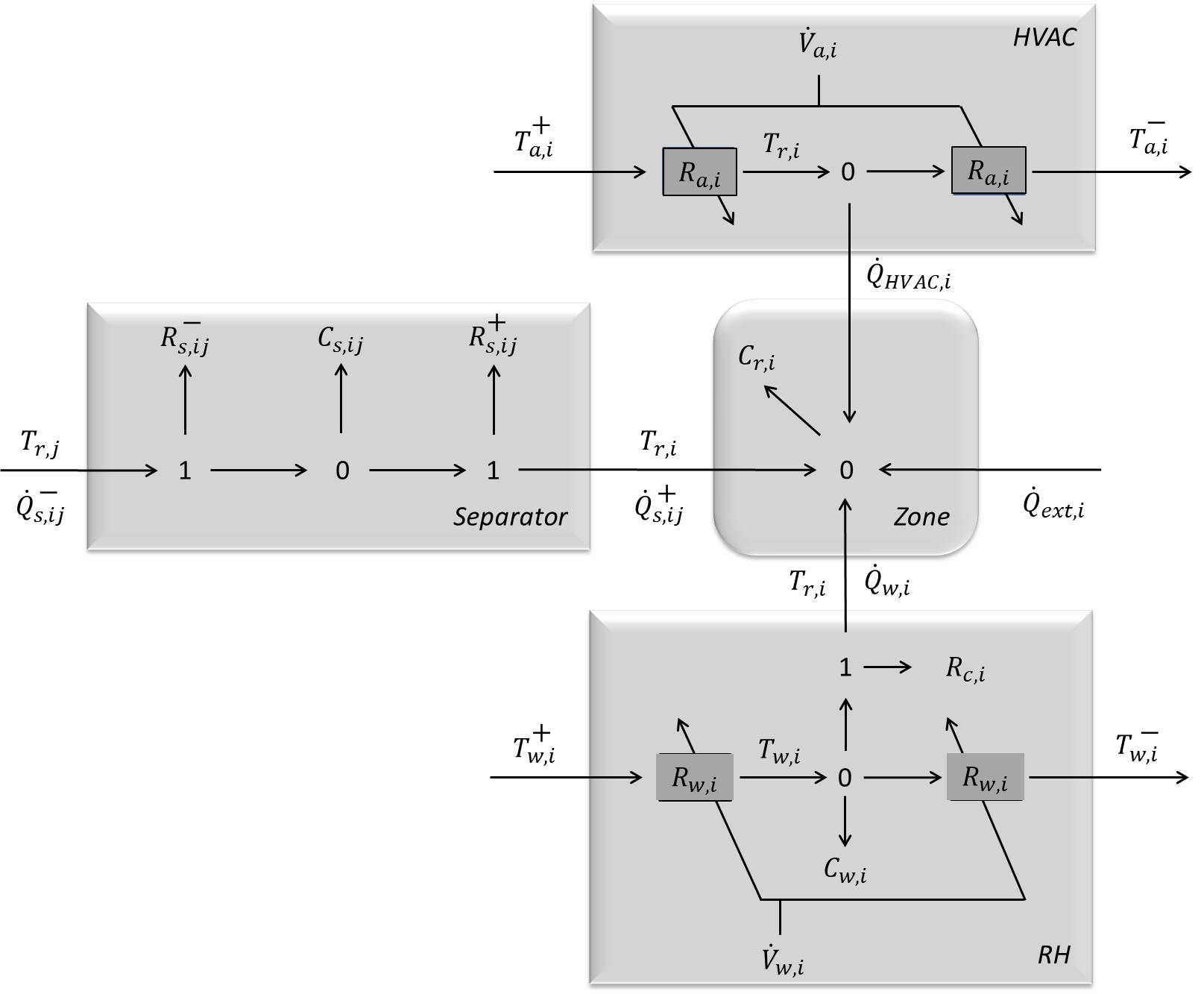}
\caption{Bond-graph approximation of heat-mass transfer dynamics of a thermal zone $i\in\mathcal{I}$.}
\label{fig:BondGraphModeling}
\end{figure*}

\subsection{Background: Heat-mass transfer dynamics}	\label{sec:HeatMassTransferDynamics}

We will use standard modeling techniques of heat-mass transfer dynamics for each thermal zone $i\in\mathcal{I}$, following a standard Bond-graph approach \cite{Karnopp12} and the simplified Fourier's law for heat conduction and Newton's law for heat convection (cf.,~\cite{Thirumaleshwar09}). Such bond-graph approximations of thermal dynamics are equivalent to a series of electrical resistance-capacitor (RC) connections. Figure~\ref{fig:BondGraphModeling} provides a detailed bond-graph approximation for a thermal zone $i\in\mathcal{I}$ with both RH and HVAC heating equipment.

We use $C$ to denote the \emph{thermal capacitance} of a separator (e.g., wall) or a thermal zone (e.g., room), and the symbol $R$ to denote the \emph{heat resistance}, according to an electricity equivalent terminology. We also use the symbol $\dot{Q}$ to denote \emph{heat flow} and the symbol $T$ to denote \emph{temperature}. 

The subscript $s$ denotes \textit{separator}-related parameters; the subscript $r$ denotes \emph{zone}-related parameters, and usually corresponds to a room; the subscript $w$ denotes parameters related to the RH system; and the subscript $a$ denotes HVAC-related parameters. Finally, a pair $ij$ denotes the interconnection between two neighboring zones $i,j\in\mathcal{I}$. 

Separators have been modeled as two thermal resistances separated by a capacitance, with the superscript ``$+$'' denoting the ``inside'' part of the separator, and the superscript ``$-$'' denoting the ``outside'' part of the separator. For example, $R_{s,ij}^{+}$ denotes the thermal resistance that it is adjacent to $i$. Furthermore, $T_{r,i}$, $i\in\mathcal{I}$, denotes the temperature of zone $i$ and $T_{s,ij}$, $j\in\mathcal{N}_i$, denotes the (internal) temperature of the separator $ij$. The separator exchanges heat with the neighboring zones $i$ and $j$ due to heat convection, whose resistance has been incorporated within $R_{s,ij}^{+}$ and $R_{s,ij}^{-}$. The heat transfer dynamics for a separator $ij$ can be written as
\begin{eqnarray}	\label{eq:SeparatorDynamics}
\dot{Q}_{s,ij}^{+} & = & \frac{1}{R_{s,ij}^{+}}(T_{s,ij} - T_{r,i}) \cr
\dot{Q}_{s,ij}^{-} & = & \frac{1}{R_{s,ij}^{-}}(T_{r,j} - T_{s,ij}) \cr
C_{s,ij}\dot{T}_{s,ij} & = & \dot{Q}_{s,ij}^{-} - \dot{Q}_{s,ij}^{+},
\end{eqnarray}
where $\dot{Q}_{s,ij}^-$ and $\dot{Q}_{s,ij}^+$ denote the heat input and output for separator $ij$, respectively. 

Regarding the RH system, it is modeled by two thermal restrictors $R_{w,i}(\dot{V}_{w,i})$ 
separated by a thermal capacitance $C_{w,i}$ (that determines the heat stored in the thermal medium). In particular, $R_{w,i}=R_{w,i}(\dot{V}_{w,i}) \df \nicefrac{1}{\rho_w c_w \dot{V}_{w,i}}$ and $C_{w,i} = c_w\rho_w V_{w}$, where $\rho_w$, $c_w$ and $V_{w}$ are the density, specific heat capacity and volume of the water in the RH system. Due to the capacitance of the thermal medium, a state variable $T_{w,i}$ is introduced that represents the temperature of the water at the point of heat exchange. The radiator exchanges heat with the thermal zone due to heat convection of resistance $R_{c,i}$. Furthermore, the inlet water temperature in the RH system is denoted by $T_{w,i}^+$, while the outlet water temperature, denoted $T_{w,i}^{-}$, is considered approximately equal to $T_{w,i}$ (if we assume uniform temperature distribution within the radiator device). Thus, the heat transfer dynamics of the RH system are
\begin{eqnarray}	\label{eq:RHDynamics}
\dot{Q}_{w,i} & = & \frac{1}{R_{c,i}}(T_{w,i}-T_{r,i}) \cr
C_{w,i}\dot{T}_{w,i} & = & \frac{1}{R_{w,i}(\dot{V}_{w,i})}(T_{w,i}^{+}-T_{w,i}) - \dot{Q}_{w,i},
\end{eqnarray}
where $\dot{Q}_{w,i}$ is the heat input to zone $i$ through the radiator.

Regarding the HVAC system, we consider the standard assumption that the fluid is incompressible, thus we are only interested in the thermal part of the system \cite{Karnopp78}. A thermal resistance, $R_{a,i}(\dot{V}_{a,i})$, is introduced to represent a restrictor of heat flow. It is given by $R_{a,i}=R_{a,i}(\dot{V}_{a,i}) = \nicefrac{1}{\rho_a c_a \dot{V}_{a,i}}$ where $\rho_a$ denotes the air density, $c_a$ denotes the air specific heat capacity at constant pressure, and $\dot{V}_{a,i}$ is the air volume rate. The outlet temperature of air in the HVAC system is considered approximately equal to the zone temperature, while the inlet air temperature is denoted by $T_{a,i}^{+}$. Thus, the heat input attributed to the HVAC system is
\begin{eqnarray}	\label{eq:HVACDynamics}
\dot{Q}_{HVAC,i} = \frac{1}{R_{a,i}(\dot{V}_{a,i})}(T_{a,i}^{+} - T_{r,i}).
\end{eqnarray}

The disturbance $\dot{Q}_{{\rm ext},i}$ denotes heat inputs from external sources. It will correspond to a vector of external heat flow rates, attributed to solar radiation, human presence, etc.

Finally, the heat-mass transfer dynamics of the thermal zone $i$ can be expressed as follows
\begin{eqnarray}	\label{eq:ZoneDynamics}
C_{r,i}\dot{T}_{r,i} = \sum_{j\in\mathcal{N}_i}\dot{Q}_{s,ij}^{+} + \dot{Q}_{w,i} + \dot{Q}_{HVAC,i} + \dot{Q}_{{\rm ext},i}.
\end{eqnarray}

From equations (\ref{eq:SeparatorDynamics}), (\ref{eq:RHDynamics}), (\ref{eq:HVACDynamics}) and (\ref{eq:ZoneDynamics}), we can derive in a straightforward manner the state-space representation of the overall system, which can be written in a generalized form as follows
\begin{eqnarray}	\label{eq:OverallSystemDynamics}
\dot{x}_i(t) = A_i(u_i(t))x_i(t) + E_i(u_i(t)) d_i(t), \quad i\in\mathcal{I},
\end{eqnarray}
where 
\begin{align*}
x_i \df & \left(\begin{array}{c}
T_{r,i} \\
{\rm col}\{T_{s,ij}\}_{j\in\mathcal{N}_i} \\
T_{w,i}
\end{array}\right), \quad u_i \df \left(\begin{array}{c}
\dot{V}_{w,i} \\ \dot{V}_{a,i}
\end{array}\right), \cr
d_i \df & \left(\begin{array}{c}
T_{w,i}^{+} \\ T_{a,i}^{+} \\ {\rm col}\{T_{r,j}\}_{j\in\mathcal{N}_i} \\ \dot{Q}_{{\rm ext},i}
\end{array}\right), \cr
\end{align*}
are the \textit{state vector}, \textit{control-input vector} and \textit{disturbance vector}, respectively. Furthermore, we define 
\begin{eqnarray*}
\begin{array}{l}
A_i(u_i) \df \cr \left[\begin{array}{ccc}
a_{r,i}(\dot{V}_{a,i}) & {\rm row}\{a_{rs,ij}^{+}\}_{j\in\mathcal{N}_i} & a_{rw,i} \cr 
{\rm col}\{a_{s,ij}^{+}\}_{j\in\mathcal{N}_i} & {\rm diag}\{a_{s,ij}\}_{j\in\mathcal{N}_i} & 0 \cr
a_{w,i} & 0 & a_{wc,i}(\dot{V}_{w,i})
\end{array}\right]
\end{array}
\end{eqnarray*}
\begin{eqnarray*}
\begin{array}{l}
E_i(u_i) \df \cr \left[\begin{array}{cccc}
0 & a_{ra,i}(\dot{V}_{a,i}) & 0 & a_{{\rm ext},i}\cr
0 & 0 & {\rm diag}\{a_{s,ij}^{-}\}_{j\in\mathcal{N}_i} & 0 \cr
a_{ww,i}(\dot{V}_{w,i}) & 0 & 0 & 0
\end{array}\right]
\end{array}
\end{eqnarray*}
where
\begin{eqnarray}	\label{eq:HeatTransferParameters}
\begin{array}{c}
a_{r,i}(\dot{V}_{a,i}) \df -\sum_{j\in\mathcal{N}_i}\nicefrac{1}{C_{r,i}R_{s,ij}^{+}} - \nicefrac{1}{C_{r,i}R_{c,i}} - \nicefrac{1}{C_{r,i}R_{a,i}(\dot{V}_{a,i})}, \cr
a_{rs,ij}^{+} \df \nicefrac{1}{C_{r,i}R_{s,ij}^{+}}, \quad a_{ra,i}(\dot{V}_{a,i}) \df \nicefrac{1}{C_{r,i}R_{a,i}(\dot{V}_{a,i})}, \cr 
a_{s,ij}^{+} \df \nicefrac{1}{C_{s,ij}R_{s,ij}^{+}}, \quad
a_{s,ij}^{-} \df \nicefrac{1}{C_{s,ij}R_{s,ij}^{-}},  \cr
a_{s,ij} \df - a_{s,ij}^{+} - a_{s,ij}^{-}, \quad 
a_{ww,i}(\dot{V}_{w,i}) \df \nicefrac{1}{C_{w,i}R_{w,i}(\dot{V}_{w,i})}, \cr
a_{w,i} \df \nicefrac{1}{C_{w,i}R_{c,i}}, \quad a_{{\rm ext},i} \df \nicefrac{1}{C_{r,i}}, \quad
a_{rw,i} \df \nicefrac{1}{C_{r,i}R_{c,i}}, \cr a_{wc,i}(\dot{V}_{w,i}) \df - a_{w,i} - a_{ww,i}(\dot{V}_{w,i}).
\end{array}
\end{eqnarray}

\subsection{Information structures}		\label{sec:InformationStructures}

\begin{table*}[tbh!]
\centering
\begin{tabular}{|c||c|c|c||c|c||c|c||c|c|}
\hline
& $T_{r,i}$ & $\{T_{s,ij}\}_{j\in\mathcal{N}_i}$ & $T_{w,i}$ & $\dot{V}_{w,i}$ & $\dot{V}_{a,i}$ & $T_{w,i}^{+}$ & $T_{a,i}^{+}$ & $\{T_{r,j}\}_{j\in\mathcal{N}_i}$ & $\dot{Q}_{{\rm ext},i}$ \\\hline\hline
(FI) & $\checkmark$ & $\times$ & $\checkmark$ & $\checkmark$ & $\checkmark$ & $\checkmark$ & $\checkmark$ & $\checkmark$ & $\checkmark$ \\\hline
(MI) & $\checkmark$ & $\times$ & $\times$ & $\checkmark$ & $\checkmark$ & $\checkmark$ & $\checkmark$ & $\checkmark$ & $\checkmark$ \\\hline
(LI) & $\checkmark$ & $\times$ & $\times$ & $\checkmark$ & $\checkmark$ & $\times$ & $\times$ & $\checkmark$ & $\checkmark$ \\\hline
\end{tabular}
\caption{Information structures with measured variables.}	
\label{Tb:ConsideredInformationStructures}
\end{table*}

Different \textit{information structures} will be considered depending on the sensors available. More specifically, \textit{\textbf{for the remainder of the paper}}, we consider the following cases:
\begin{itemize}
\item[$\bullet$] \textit{Full information (FI):} All state, control and disturbance variables can be measured except for the separator state variables $\{T_{s,ij}\}_{j}$.
\item[$\bullet$] \textit{Medium Information (MI):} All state, control and disturbance variables can be measured except for the separator state variables, $\{T_{s,ij}\}_{j}$, and the water state variable at the point of heat exchange, $T_{w,i}$.
\item[$\bullet$] \textit{Limited information (LI):} All state, control and disturbance variables can be measured except for the separator state variables $\{T_{s,ij}\}_{j}$, the water state variable $T_{w,i}$ and the disturbance variables $T_{w,i}^{+}$ and $T_{w,i}^{-}$. 
\end{itemize}

The specifics can also be visualized in Table~\ref{Tb:ConsideredInformationStructures}. The introduced alternative structures intend on demonstrating the flexibility of the proposed derivation of nonlinear regressors, since different residential buildings may not necessarily incorporate all sensors of the FI structure.

The case of limited information (LI) corresponds to the case where only the flow rates are known in the RH and/or the HVAC system, while the corresponding inlet/outlet temperatures of the thermal medium cannot be measured. Such assumption may be reasonable for most today's residential buildings.

Alternative information structures may be considered, such as the ones where $\dot{Q}_{{\rm ext},i}$ (due to radiation or peoples' presence) is not known, however such cases can be handled through the above cases and with small modifications.

\subsection{System separation under FI}	\label{sec:SystemSeparationUnderFI}

Under the full information structure (FI), a natural separation of the system dynamics can be introduced between the zone and RH dynamics. Such separation will be used at several points throughout the paper, and it is schematically presented in Figure~\ref{fig:SystemSeparationUnderFI}. 

\begin{figure}[h!]
\centering
\includegraphics{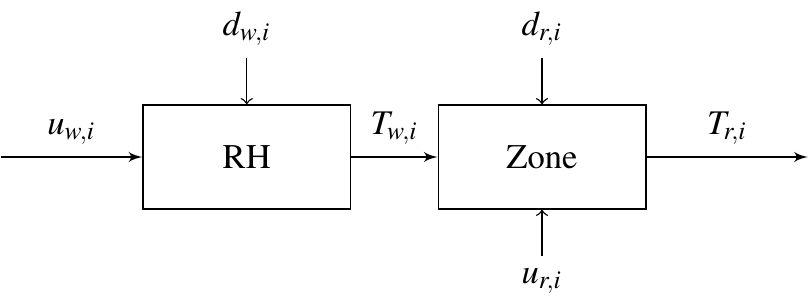}
\caption{System separation architecture under FI.}
\label{fig:SystemSeparationUnderFI}
\end{figure}

In particular, the overall (heat-mass transfer) dynamics for a single zone can be written as follows
\begin{eqnarray}	\label{eq:SystemSeparationUnderFI:RoomDynamics}
\dot{x}_{r,i}(t) = A_{r,i}(u_{r,i}(t)) x_{r,i}(t) + E_{r,i}(u_{r,i}(t)) d_{r,i}(t)
\end{eqnarray}
where $u_{r,i} \df \dot{V}_{a,i}$ is the control-input vector and
\begin{eqnarray*}
x_{r,i} \df \left(\begin{array}{c}
T_{r,i} \\
{\rm col}\{T_{s,ij}\}_{j\in\mathcal{N}_i} 
\end{array}\right), \quad d_{r,i} \df  \left(\begin{array}{c}
T_{w,i} \\ T_{a,i}^{+} \\ {\rm col}\{T_{r,j}\}_{j\in\mathcal{N}_i} \\ \dot{Q}_{{\rm ext},i}
\end{array}\right),
\end{eqnarray*}
are the state vector and disturbance vector of the zone dynamics, respectively. Furthermore, we define
\begin{eqnarray*}
A_{r,i}(u_i) \df \left[\begin{array}{cc}
a_{r,i}(\dot{V}_{a,i}) & {\rm row}\{a_{rs,ij}^{+}\}_{j\in\mathcal{N}_i} \cr 
{\rm col}\{a_{s,ij}^{+}\}_{j\in\mathcal{N}_i} & {\rm diag}\{a_{s,ij}\}_{j\in\mathcal{N}_i}
\end{array}\right], \cr
\end{eqnarray*}
\begin{eqnarray*}
E_{r,i} \df \left[\begin{array}{cccc}
a_{rw,i} & a_{ra,i}(\dot{V}_{a,i}) & 0 & a_{{\rm ext},i}  \cr
0 & 0 & {\rm diag}\{a_{s,ij}^{-}\}_{j\in\mathcal{N}_i} & 0
\end{array}\right].
\end{eqnarray*}

The RH dynamics can be written as follows
\begin{eqnarray}	\label{eq:SystemSeparationUnderFI:WaterDynamics}
\dot{x}_{w,i}(t) = A_{w,i}(u_{w,i}(t)) x_{w,i}(t) + E_{w,i}(u_{w,i}(t)) d_{w,i}(t),
\end{eqnarray}
where
\begin{eqnarray*}
x_{w,i} \df T_{w,i}, \quad u_{w,i} \df \dot{V}_{w,i},
\quad d_{w,i} \df \left(\begin{array}{c}
T_{w,i}^{+} \\ T_{r,i}
\end{array}\right)
\end{eqnarray*}
are the state vector, control-input vector and disturbance vector of the RH system, respectively. Furthermore, we define the following system matrices
\begin{eqnarray*}
A_{w,i}(u_i) \df \left[a_{wc,i}(\dot{V}_{w,i})\right], \quad 
E_{w,i} \df \left[\begin{array}{cc} a_{ww,i}(\dot{V}_{w,i}) & a_{w,i} \end{array}\right].
\end{eqnarray*}


\section{Linear Regression Models (LRM)}	\label{sec:LinearRegressionModels}

We would like to focus on models for system identification which (to great extent) retain the physical interpretation of the parameters. To this end, an output-error model structure (cf.,~\cite{Ljung99}) is considered, where a one-to-one correspondence between the identification parameters and the discrete-time model parameters can be established. In the remainder of this section, we provide some background on the Output-Error (OE) model structure, while we also discuss the implicit assumptions made when a linear structure of the dynamics is assumed.

\subsection{Background: Output-Error (OE) model}		\label{sec:OutputErrorModel}

If we assume that the relation between input $u$ and the undisturbed output of the system $w$ can be written as a linear difference equation, and that the disturbance consists of white measurement noise $v$, then we obtain:
\begin{eqnarray}	\label{eq:InputOutputDescription}
\lefteqn{w(k) + f_1w(k-1) + ... + f_{n_f}w(k-n_f)}\cr && = b_1u(k-1) + ... + b_{n_b}u(k-n_b) 
\end{eqnarray}
and the output is $y(k) = w(k) + v(k)$, for some positive integers $n_f\geq n_b$.
The parameter vector to be determined is: $\theta\df (b_1,\ldots,b_{n_b},f_1,\ldots,f_{n_f}).$
Since $w(k)$ is not observed, it should be constructed from previous inputs and it should carry an index $\theta$. 

The natural predictor (resulting from a \textit{maximum a posteriori} predictor\footnote{i.e., the one that maximizes the probability density function of the output given observations up to the previous time instant.}) is:
$\hat{y}(k|\theta) = w(k,\theta),$ 
and it is constructed from past inputs only. If we define the vector:
\begin{eqnarray}	\label{eq:RegressionVector}
\lefteqn{\varphi(k,\theta) \df \Big(u(k-1),\ldots, u(k-n_b), } \cr && -w(k-1,\theta), \ldots, -w(k-n_f,\theta) \Big),
\end{eqnarray}
then, the predictor can be written more compactly as:
$\hat{y}(k|\theta) = \varphi(k,\theta)\tr \theta,$ leading to a \emph{linear regression model} (LRM), where $\varphi(k,\theta)$ is called the \emph{regression vector}. To simplify notation, in several cases we will write $\varphi(k)$ instead of $\varphi(k,\theta)$.
Note that in the above expression, the $w(k-j,\theta)$, $j=1,2,...,n_f$, are not observed, but using the above maximum a-posteriori predictor, they can be computed using previous inputs as follows:
$w(k-j,\theta) = \hat{y}(k-j|\theta),$ $j=1,2,...,n_f.$

Such output-error model predictors will be used throughout the paper. However, note that, depending on the physics of the identified system, the regression vector $\varphi$ defined in (\ref{eq:RegressionVector}) may be nonlinear with respect to the input/output lagged variables. In such cases, the resulting predictor will be referred to as a \emph{linear regression model with nonlinear regressors} (NRM).

\subsection{Discussion}

Output-error model structures do not consider any disturbance terms in the process dynamics, instead they only consider measurement noise. Hence, such model structures are rather appropriate for providing a clear picture of the impact of the assumed (physical) model into the prediction error. On the other hand, when considering structures with process noise, such as ARMAX model structures (cf.,~\cite[Section~4.2]{Ljung99}), the process dynamics are perturbed by artificial terms, which are not easily motivated by the physical insight. Since the goal of this paper is to evaluate the impact of the assumed (physical) dynamics, we consider OE models more appropriate for evaluating predictions. This, however, does not imply that the performance of an OE model structure is necessarily better compared to an ARMAX model. The goal of this paper is \emph{not} to provide such a comparison.

Furthermore, \emph{identifiability} (cf.~\cite[Section~4.6]{Ljung99}) of the considered OE model structures will be guaranteed by the (inherent) controllability of the heat-mass transfer dynamics.

\subsection{Linear approximation \& implicit assumptions}		\label{sec:OutputErrorModel}

We wish to explore the utility of an OE model structure using a linear regression vector $\varphi$ into identifying the nonlinear system dynamics of Equation~(\ref{eq:OverallSystemDynamics}). Note that the original dynamics (\ref{eq:OverallSystemDynamics}) are bilinear in nature due to multiplications of the flow rates ($\dot{V}_{w,i}$, $\dot{V}_{a,i}$) with state or disturbance variables. 

An investigation of the original dynamics (\ref{eq:OverallSystemDynamics}) reveals that when an output-error model is considered that admits a linear regression vector $\varphi$, we implicitly admit assumptions that may lead to significant identification errors. To see this, let us consider the full information structure (FI). In a forthcoming section, we will show that the following approximation of the RH part of the dynamics holds for sufficiently small sampling period $\epsilon>0$:
\begin{eqnarray*}	\label{eq:LinearApproximation:NoHVAC}
\lefteqn{T_{w,i}(k+1) \approx } \cr && (1 + \epsilon a_{wc,i}(\dot{V}_{w,i}(k))) T_{w,i}(k) + 
\epsilon a_{ww,i}(\dot{V}_{w,i}(k)) T_{w,i}^{+}(k) + \cr && \epsilon a_{w,i} T_{r,i}(k),
\end{eqnarray*}
plus higher-order terms of $\epsilon$. 

According to this (finite-response) approximation, a linear regression vector $\varphi$ may describe well the evolution of $T_{w,i}$ as long as either the flow rate $\dot{V}_{w,i}$ or the temperatures $T_{w,i}$ and $T_{w,i}^{+}$ are \emph{not} varying significantly with time. However, variations in the flow rate $\dot{V}_{w,i}$ and in the water temperature $T_{w,i}$ may be large with time. Similar are the conclusions when investigating the effect of the air flow rate $\dot{V}_{a,i}$ in the evolution of the zone temperature. The exact effect of these nonlinear effects \emph{cannot} be a-priori determined.

\section{Linear Regression Models with Nonlinear Regressors (NRM)}		\label{sec:NonLinearRegressionModels}

In this section, we will explore the formulation of OE model structures when the regressors in $\varphi$ may be nonlinear functions of lagged input/output variables. Such investigation will be performed under the FI, MI and LI structure, extending prior work of the same authors \cite{ChasparisNatschlaeger14} to a larger set of possible information structures (beyond the FI structure).

For the derivation, we will be using the following notation: Let $\qq^{-1}\{\cdot\}$ denote the one-step delay operator, i.e., $\qq^{-1}\{x(k)\}=x(k-1)$. Note that the delay operator is \emph{linear}. 
Let us also define the following operators:
\begin{itemize}
\item $\PP_{s,ij} \df [(1+\epsilon a_{s,ij})\qq^{-1}]$
\item $\PP_{wc,i}(\dot{V}_{w,i}) \df [(1+\epsilon a_{wc,i}(\dot{V}_{w,i}))\qq^{-1}]$
\item $\PP_{r,i}(\dot{V}_{a,i}) \df [(1+\epsilon a_{r,i}(\dot{V}_{a,i}))\qq^{-1}]$
\end{itemize}
where $\epsilon>0$ defines the sampling period. Define also 
\begin{itemize}
\item $\QQ_{s,ij}\df [1-\PP_{s,ij}]$, 
\item $\QQ_{wc,i}(\dot{V}_{w,i})\df [1 - \PP_{wc,i}(\dot{V}_{w,i})]$, 
\item $\QQ_{r,i}(\dot{V}_{a,i}) \df [1-\PP_{r,i}(\dot{V}_{a,i})]$.
\end{itemize}
For any operator $\PP\{\cdot\}$, $\PP^{-1}\{\cdot\}$ will denote its inverse operator, i.e., $$\PP^{-1}\{\PP\{x(k)\}\}=x(k).$$

\begin{property}	\label{Pr:Property1}
For each zone $i\in\mathcal{I}$, the family of operators $\{\QQ_{s,ij}\}_{j\in\mathcal{N}_i}$ are pairwise commutative, i.e., $$\QQ_{s,ij}\QQ_{s,ij'}\{x(k)\} = \QQ_{s,ij'}\QQ_{s,ij}\{x(k)\}$$ for any $j,j'\in\mathcal{I}$ with $j\neq j'$.
\end{property}
\begin{proof}
For each zone $i\in\mathcal{I}$, and for any two neighboring zones $j,j'\in\mathcal{N}_i$, we have the following:
\begin{eqnarray*}
\lefteqn{\QQ_{s,ij}\QQ_{s,ij'}\{x(k)\}} \cr & = &
\left[1 - \PP_{s,ij}\right]\left\{ \left[1 - \PP_{s,ij'} \right] \left\{x(k)\right\} \right\} \cr & = & 
\left[1 - \PP_{s,ij}\right]\left\{ x(k) - (1+\epsilon a_{s,ij'}) x(k-1) \right\} \cr & = & 
x(k) - (1+\epsilon a_{s,ij'}) x(k-1) - \cr && (1+\epsilon a_{s,ij})x(k-1) + (1+\epsilon a_{s,ij})(1+\epsilon a_{s,ij'})x(k-2).
\end{eqnarray*}
It is straightforward to check that the same expression also occurs if we expand $\QQ_{s,ij'}\QQ_{s,ij}$, due to the fact that $(1+\epsilon a_{s,ij})$ commutes with $(1+\epsilon a_{s,ij'})$. Thus, the conclusion follows. $\bullet$
\end{proof}

Another two properties that will be handy in several cases are the following:
\begin{property}	\label{Pr:Property2}
For each zone $i\in\mathcal{I}$, we have
\begin{eqnarray*}
\lefteqn{[\QQ_{s,ij}\QQ_{wc,i}(\dot{V}_{w,i}(k))] = [\QQ_{wc,i}(\dot{V}_{w,i}(k))\QQ_{s,ij}] - }\cr &&
(1+\epsilon a_{s,ij})\epsilon \left[(1-\qq^{-1})a_{wc,i}(\dot{V}_{w,i}(k))\right]\qq^{-2}.
\end{eqnarray*}
\end{property}
\begin{proof}
See Appendix~\ref{Ap:Property2}. $\bullet$
\end{proof}
\begin{property}	\label{Pr:Property3}
For some finite index set $A$ and signal $x(k)$,
\begin{equation}
\Big[\prod_{j\in{A}}\QQ_{s,ij}\Big]\{x(k)\} = x(k) + \sum_{m=1}^{\magn{A}}\alpha_m x(k-m),
\end{equation}
for some constants $\alpha_1,...,\alpha_{\magn{A}}$.
\end{property}
\begin{proof}
The conclusion follows in a straightforward manner by induction and the definition of the delay operator $\QQ_{s,ij}$. $\bullet$
\end{proof}

\subsection{Full-Information (FI) structure}	\label{sec:NRM:FIstructure}

Using the natural decomposition of the dynamics under the FI structure described in Section~\ref{sec:SystemSeparationUnderFI}, in the following subsections we provide a derivation (using physical insight) of nonlinear regression vectors under an OE model structure for the two subsystems of the overall dynamics.

\subsubsection{RH dynamics}		\label{sec:RadiantHeatingSystem}

By isolating the RH dynamics, we first formulate a (finite-response) prediction of the water temperature $T_{w,i}$, as the following proposition describes.
\begin{proposition}	[NRM for RH under FI]	\label{Pr:NRM:RHDynamics}
For sufficiently small sampling period $\epsilon>0$, the RH dynamics (\ref{eq:SystemSeparationUnderFI:WaterDynamics}) can be approximated by:
\begin{eqnarray}	\label{eq:FI:WaterTempApproximation}
\lefteqn{T_{w,i}(k) \approx \PP_{wc,i}(\dot{V}_{w,i}(k-1))\{T_{w,i}(k)\} +} \cr && \Big.\epsilon a_{ww,i}(\dot{V}_{w,i}(k-1))T_{w,i}^{+}(k-1) + \epsilon a_{w,i}T_{r,i}(k-1),
\end{eqnarray}
plus higher-order terms of $\epsilon$. Furthermore, the maximum a-posteriori predictor of $T_{w,i}(k)$ can be approximated by $\hat{T}_{w,i}(k|\theta_{w,i}) \approx \varphi_{w,i}(k)\tr\theta_{w,i}$ plus higher-order terms of $\epsilon$, where $\theta_{w,i}$ is a vector of unknown parameters 
and
\begin{eqnarray}	\label{eq:RHpredictorFI}
\varphi_{w,i} (k) \df 
\left(\begin{array}{c}
\hat{T}_{w,i}(k-1|\theta_{w,i}) \\
\dot{V}_{w,i}(k-1) \hat{T}_{w,i}(k-1|\theta_{w,i}) \\
\dot{V}_{w,i}(k-1) T_{w,i}^{+}(k-1) \\
T_{r,i}(k-1)
\end{array}\right).
\end{eqnarray}
\end{proposition}
\begin{proof}
By a Taylor-series expansion of the RH dynamics, the finite-step response of the water temperature $T_{w,i}(k)$ is
\begin{eqnarray*}
\lefteqn{T_{w,i}(k) \approx } \cr && (1 + \epsilon a_{wc,i}(\dot{V}_{w,i}(k-1))) T_{w,i}(k-1) + \cr && \epsilon a_{ww,i}(\dot{V}_{w,i}(k-1)) T_{w,i}^{+}(k-1) + \epsilon a_{w,i} T_{r,i}(k-1),
\end{eqnarray*}
plus higher-order terms of $\epsilon$. Equation~(\ref{eq:FI:WaterTempApproximation}) directly results from the definition of the delay operator $\PP_{wc,i}(\dot{V}_{w,i})$.
From Equation~(\ref{eq:FI:WaterTempApproximation}), a regression vector may be derived of the desired form. $\bullet$
\end{proof}

\subsubsection{Zone dynamics}

A similar approach to Proposition~\ref{Pr:NRM:RHDynamics} for the derivation of a nonlinear regression vector can also be employed for the zone dynamics. 

\begin{proposition}[NRM for Zone under FI]		\label{Pr:NRM:ZoneDynamicsFI}
For sufficiently small sampling period $\epsilon>0$, the maximum a-posteriori predictor of $T_{r,i}(k)$ can be approximated by $\hat{T}_{r,i}(k|\theta_{r,i})\approx \varphi_{r,i}(k)\tr\theta_{r,i}$ plus higher-order terms of $\epsilon$, where $\theta_{r,i}$ is a vector of unknown parameters 
and
\begin{eqnarray*}
\varphi_{r,i} (k) \df 
\left(\begin{array}{c}
{\rm col}\{\hat{T}_{r,i}(k-m|\theta_{r,i})\}_{m=1}^{\magn{\mathcal{N}_i}+1} \\
{\rm col}\Big\{{\rm col}\{T_{r,j}(k-m)\}_{m=2}^{\magn{\mathcal{N}_i}+1}\Big\}_{j\in\mathcal{N}_{i}} \\
{\rm col}\{\dot{V}_{a,i}(k-m) \hat{T}_{r,i}(k-m|\theta_{r,i})\}_{m=1}^{\magn{\mathcal{N}_i}+1} \\
{\rm col}\{\dot{V}_{a,i}(k-m) T_{a,i}^{+}(k-m)\}_{m=1}^{\magn{\mathcal{N}_i}+1} \\
{\rm col}\{\hat{T}_{w,i}(k-m|\theta_{w,i})\}_{m=1}^{\magn{\mathcal{N}_i}+1} \\
{\rm col}\{\dot{Q}_{{\rm ext},i}(k-m)\}_{m=1}^{\magn{\mathcal{N}_i}+1}
\end{array}\right).
\end{eqnarray*}
\end{proposition}
\begin{proof}
See Appendix~\ref{Ap:NRM:ZoneDynamicsFI}. $\bullet$
\end{proof}

\subsection{Medium-Information (MI) structure}

The case of MI structure is slightly different from the FI structure, since the water temperature, $T_{w,i}$, at the point of heat exchange cannot be directly measured, therefore the dynamics of the RH and the zone cannot be separated as was demonstrated in Figure~\ref{fig:SystemSeparationUnderFI}. The following proposition provides the corresponding NRM predictor for the MI structure.

\begin{proposition}[NRM for Zone under MI]		\label{Pr:NRM:ZoneDynamicsMI}
For sufficiently small sampling period $\epsilon>0$, the maximum a-posteriori predictor of $T_{r,i}(k)$ can be approximated by $\hat{T}_{r,i}(k|\theta_{r,i})\approx \varphi_{r,i}(k)\tr\theta_{r,i}$ plus higher-order terms of $\epsilon$, where $\theta_{r,i}$ is a vector of unknown parameters 
and
\begin{eqnarray}	\label{eq:NonlinearRegressorMI}
\lefteqn{\varphi_{r,i} (k) \df  } \cr &&
\left(\begin{array}{c}
{\rm col}\{\hat{T}_{r,i}(k-m|\theta_{r,i})\}_{m=1}^{\magn{\mathcal{N}_i}+2} \\
{\rm col}\Big\{{\rm col}\{T_{r,j}(k-m)\}_{m=2}^{\magn{\mathcal{N}_i}+2}\Big\}_{j\in\mathcal{N}_{i}} \\
{\rm col}\{\dot{V}_{a,i}(k-m) \hat{T}_{r,i}(k-m|\theta_{r,i})\}_{m=1}^{\magn{\mathcal{N}_i}+2} \\
{\rm col}\{\dot{V}_{a,i}(k-m) T_{a,i}^{+}(k-m)\}_{m=1}^{\magn{\mathcal{N}_i}+2} \\
{\rm col}\{\dot{V}_{a,i}(k-m) \dot{V}_{w,i}(k-m) T_{a,i}^{+}(k-m)\}_{m=2}^{\magn{\mathcal{N}_i}+2} \\
{\rm col}\{\dot{V}_{w,i}(k-m-1)\hat{T}_{r,i}(k-m|\theta_{r,i})\}_{m=1}^{\magn{\mathcal{N}_i}+1}\\
{\rm col}\{\dot{V}_{w,i}(k-m)\hat{T}_{r,i}(k-m|\theta_{r,i})\}_{m=2}^{\magn{\mathcal{N}_i}+2}\\
{\rm col}\{\dot{V}_{w,i}(k-m)\dot{V}_{a,i}(k-m)\hat{T}_{r,i}(k-m|\theta_{r,i})\}_{m=2}^{\magn{\mathcal{N}_i}+2}\\
{\rm col}\{\dot{V}_{w,i}(k-m)T_{w,i}^{+}(k-m)\}_{m=2}^{\magn{\mathcal{N}_i}+2}\\
{\rm col}\{\dot{Q}_{{\rm ext},i}(k-m)\}_{m=1}^{\magn{\mathcal{N}_i}+2} \\
{\rm col}\{\dot{V}_{w,i}(k-m)\dot{Q}_{{\rm ext},i}(k-m)\}_{m=2}^{\magn{\mathcal{N}_i}+2} \\
\end{array}\right).
\end{eqnarray}
\end{proposition}
\begin{proof}
The proof follows similar reasoning with the proof of Proposition~\ref{Pr:NRM:ZoneDynamicsFI} and it is presented in Appendix~\ref{Ap:NRM:ZoneDynamicsMI}. $\bullet$
\end{proof}

\subsection{Limited-Information (LI) structure}

Similarly to the derivation of the NRM for the FI and MI structures, we may also derive the corresponding NRM for the case of the LI structure. Note that the difference of such information structure (compared to MI) is the fact that the inlet temperatures of the thermal mediums ($T_{w,i}^+$ and $T_{a,i}^+$) cannot be measured. In such case, the only valid assumption is the fact that the inlet temperature of the thermal mediums remain constant with time. In this case, Proposition~\ref{Pr:NRM:ZoneDynamicsMI} continue to hold, when $T_{w,i}^{+}$ and $T_{a,i}^{+}$ are replaced with unity in (\ref{eq:NonlinearRegressorMI}) (i.e., they become part of the parameters in $\theta_{r,i}$).

\section{Performance Evaluation}	\label{sec:PerformanceEvaluation}

In this section, we provide a comparison between standard linear regression models and the nonlinear ones derived in Section~\ref{sec:NonLinearRegressionModels}. The comparison will be performed both with respect to a) \emph{identification error}, and b) \emph{prediction performance} within a standard MPC implementation for climate control in a residential building. The second part is considered the most representative of the performance of an identification model since even small prediction errors might lead to significant performance degradation.

\subsection{Simulation platform}

In order to compare the proposed nonlinear regression vectors with standard linear regressors, we used the EnergyPlus (V7-2-0) building simulator developed by the U.S. Department of Energy \cite{EnergyPlus}. The BCVTB simulation tool has also been used for allowing data collection and also climate control developed in MATLAB to be implemented during run-time. A typical (three-storey) residential building in Linz, Austria, was accurately modeled and simulated with the EnergyPlus environment to allow for collecting data from a realistic residential environment.

\subsection{Data generation}		\label{sec:DataGeneration}

The data collection for system identification was performed under normal operating conditions of the heating system during the winter months (October - April) under the weather conditions of Linz, Austria. 
%
%
To replicate normal operating conditions, a standard hysteresis controller was employed, according to which the water flow $\dot{V}_{w,i}(k)$ is updated regularly at time instances $t_k=kT_{\rm sam}$ where $T_{\rm sam}=\nicefrac{1}{12}h$. The control law is as follows: 
\begin{eqnarray*}
\lefteqn{\dot{V}_{w,i}(k) \df} \cr && \begin{cases}
\dot{V}_{w,{\rm max}} & \mbox{if } p_i(k) > 0 \mbox{ and } \\ 
  &  (T_{r,i}(k) < T_{\rm set} - \Delta{T} \mbox{ or }\\
  & ( T_{r,i}(k) \geq T_{\rm set} - \Delta{T} \mbox{ and } T_{r,i}(k-1) \leq T_{r,i}(k) ) ),\\
0 & \mbox{else }
\end{cases}
\end{eqnarray*}
where $\Delta{T}$ determines a small temperature range about the desired (set) temperature, $T_{\rm set}$, in which the control flow maintains its value and $\dot{V}_{w,{\rm max}}$ denotes the maximum water flow set to $0.0787kg/sec$.  
Furthermore, $p_i(k)\in\{0,1\}$ indicates whether people are present in thermal zone $i$. In other words, the hysteresis controller applies only if someone is present in the thermal zone $i$, which can be easily determined through a motion detector. Furthermore, the inlet water temperature $T_{w,i}^{+}$ is determined by the following heating curve:
\begin{equation}	\label{eq:HeatingCurve}
T_{w,i}^{+} \df \begin{cases}
\rho_0 + \rho_1\cdot (T_{\rm set} - T_{\rm out})^\zeta, & \mbox{if } T_{\rm set} > T_{\rm out} \\
\rho_0, & \mbox{else.}
\end{cases}
\end{equation}
where $\rho_0$, $\rho_1$ and $\zeta$ are positive constants. For the simulation, we set $\Delta{T}=0.1$, $\rho_0=29.30$, $\rho_1=0.80$ and $\zeta=0.97$. The set temperature was set equal to $T_{\rm set}=21^oC$. 

There exists a natural ventilation system that operates autonomously with an intermittent flow pattern $\dot{V}_{a,i}$. Thus, the only control parameter is $\dot{V}_{w,i}$. All parameters mentioned in Table~\ref{Tb:ConsideredInformationStructures} can be measured, except for $\{T_{s,ij}\}_{j}$, which allows for evaluating all considered structures (FI, MI and LI).

%

\subsection{Recursive system identification}	\label{sec:RecursiveSystemIdentification}

To utilize the linear and nonlinear regression vectors for system identification of heat-mass transfer dynamics of a residential building, we use an OE model structure as described in Section~\ref{sec:NonLinearRegressionModels}, while we employ a regularized recursive least-squares implementation (cf.,~\cite[Section~12.3]{Sayed03}) for training its parameters. The reason for implementing a recursive identification procedure is primarily due to the fact that the size of the available data is quite large, which makes the use of standard least-squares approaches practically infeasible. Besides, predictions for the zone temperature will be  needed continuously during run-time, demanding for more efficient computational schemes. Furthermore, a recursive least squares implementation allows for an adaptive response to more recent data, thus capturing more recent effects.

\subsection{Identification experiments}					\label{sec:PersistenceExcitation}

The experiments performed in this study will involve a) a standard linear regression model, with b) the derived NRM of Proposition~\ref{Pr:NRM:ZoneDynamicsMI} (which corresponds to the MI structure). Our intention is to evaluate the benefit of the more accurate physical representation of the models derived in Section~\ref{sec:NonLinearRegressionModels}.
In both regression vectors (linear or nonlinear) we consider $|\mathcal{N}_i|=1$, i.e., we implicitly assume that there is a single neighboring zone. This can be justified by the fact that the building can be considered as a single thermal zone, since the same heating schedule is employed in all rooms.

In particular, the linear regression model (LRM) implemented corresponds to the following structure:
\begin{eqnarray}	\label{eq:LinearRegressorMI}
\varphi_{r,i} (k) \df  
\left(\begin{array}{c}
{\rm col}\{\hat{T}_{r,i}(k-m|\theta_{r,i})\}_{m=1}^{\magn{\mathcal{N}_i}+2} \\
{\rm col}\Big\{{\rm col}\{T_{r,j}(k-m)\}_{m=1}^{\magn{\mathcal{N}_i}+2}\Big\}_{j\in\mathcal{N}_{i}} \\
{\rm col}\{\dot{V}_{a,i}(k-m)\}_{m=1}^{\magn{\mathcal{N}_i}+2} \\
{\rm col}\{T_{a,i}^{+}(k-m)\}_{m=1}^{\magn{\mathcal{N}_i}+2} \\
{\rm col}\{\dot{V}_{w,i}(k-m)\}_{m=1}^{\magn{\mathcal{N}_i}+2} \\
{\rm col}\{T_{w,i}^{+}(k-m|\theta_{r,i})\}_{m=1}^{\magn{\mathcal{N}_i}+2}\\
{\rm col}\{\dot{Q}_{{\rm ext},i}(k-m)\}_{m=1}^{\magn{\mathcal{N}_i}+2} \\
\end{array}\right),
\end{eqnarray}
i.e., it only considers delayed samples of the state and the available input and disturbance parameters. The NRM model implemented corresponds to the one of Equation~(\ref{eq:NonlinearRegressorMI}).

\subsection{Persistence of Excitation}

Following the analysis of Section~\ref{sec:NonLinearRegressionModels}, for any thermal zone $i\in\mathcal{I}$, the output of the zone $y_i(k)\df T_{r,i}(k)$ can briefly be expressed in the following form:
\begin{equation}	\label{eq:PersistenceExcitation:BasicModel}
y_i(k) = G_{i0}(\qq)y_i(k) + \sum_{l\in{L}}G_{il}(\qq)z_{il}(k),
\end{equation}
for some linear transfer functions $G_{i0}(\qq)$, $G_{il}(\qq)$, $l\in{L}$ and some index set $L\subset\mathbb{N}$. The terms $z_{il}(k)$ represent time sequences of input, disturbance or output signals, or products of such terms. For example, if we consider the MI model presented in (\ref{eq:NonlinearRegressorMI}), the terms $z_{il}$ may correspond to a disturbance, such as $T_{r,j}$, $j\in\mathcal{N}_i$, or to a product of a control input and the output, such as $\dot{V}_{w,i}(k)T_{r,i}(k)$. 

A model of the generic form of equation (\ref{eq:PersistenceExcitation:BasicModel}) is uniquely determined by its corresponding transfer functions $G_{i0}(\qq)$, $G_{il}(\qq)$, $l\in{L}$. Thus, in order for the data to be informative, the data should be such that, for \emph{any} two different models of the form (\ref{eq:PersistenceExcitation:BasicModel}) the following condition is satisfied:\footnote{cf., discussion in \cite[Section~13.4]{Ljung99}.}
\begin{equation}	\label{eq:InformativeExperimentsCondition}
\overline{E}\left[\Delta{G}_{i0}(\qq) y_i(k) + \sum_{l\in{L}}\Delta{G}_{il}(\qq)z_{l}(k)\right]^{2} \neq 0,
\end{equation}
where $\Delta G_{i0}(\qq)$, $\Delta G_{il}(\qq)$ denote the corresponding differences between the transfer functions of the two different models. Moreover, $\overline{E}\df\lim_{N\to\infty}\nicefrac{1}{N}\sum_{k=1}^{N}E[\cdot]$ and $E[\cdot]$ denotes the expectation operator. Since the models are assumed different, either $\Delta G_{i0}(\qq)$ or $\Delta G_{il}(\qq)$, for some $l\in{L}$, should be nonzero. Thus, if $\Delta G_{il}(\qq)\neq{0}$, for some $l\in{L}$, the \emph{persistence of excitation} \footnote{cf.,~\cite[Definition~13.2]{Ljung99}} of $z_{il}(k)$ will guarantee (\ref{eq:InformativeExperimentsCondition}). 

Since data are collected through a closed-loop experiment, special attention is needed for terms $z_{il}(k)$ which involve the product of the water flow rate $\dot{V}_{w,i}(k)$ with the output $y_i(k)$ of the system, such as the term $\dot{V}_{w,i}(k)y_{i}(k)$ in (\ref{eq:NonlinearRegressorMI}). Note, however, that this term is \emph{nonlinear} in both $y_i(k)$ and the presence/occupancy indicator, since the hysteresis controller a) applies only when someone is present in the zone and b) is a nonlinear function of the output $y_i(k)$. Thus, it is sufficient for either the input signal or the occupancy indicator to be persistently exciting in order for the experiment to be informative. In Figure~\ref{fig:Spectrum}, we provide the positive spectrum generated through the discrete Fast Fourier Transform (DFFT) of the disturbance and control signals.

\begin{figure}[h!]
\centering
\includegraphics{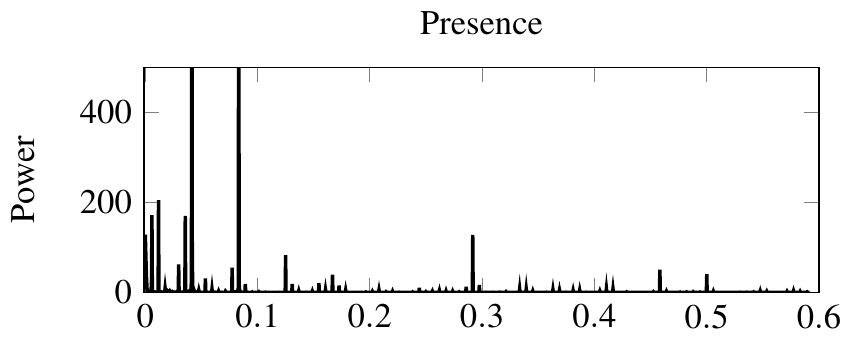}
\includegraphics{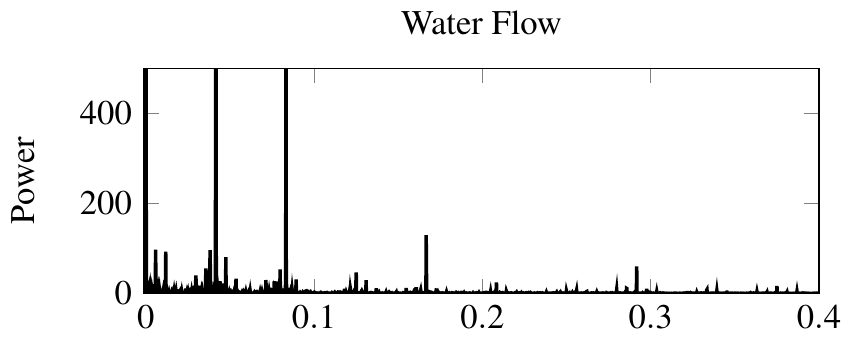}
\includegraphics{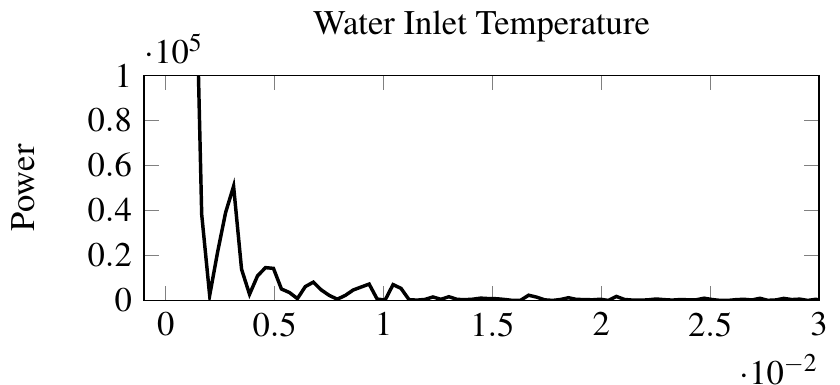}
\includegraphics{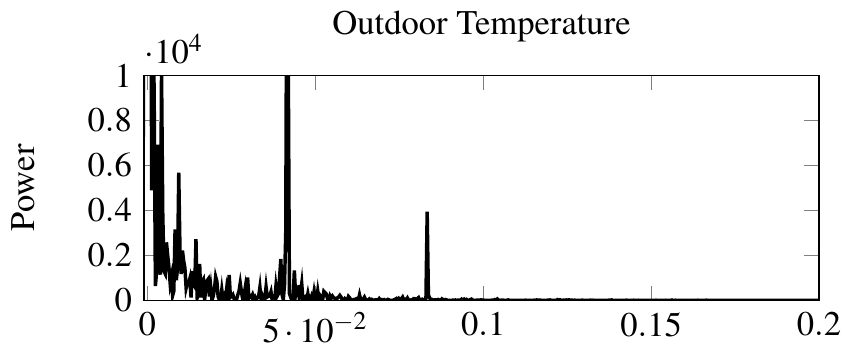}
\includegraphics{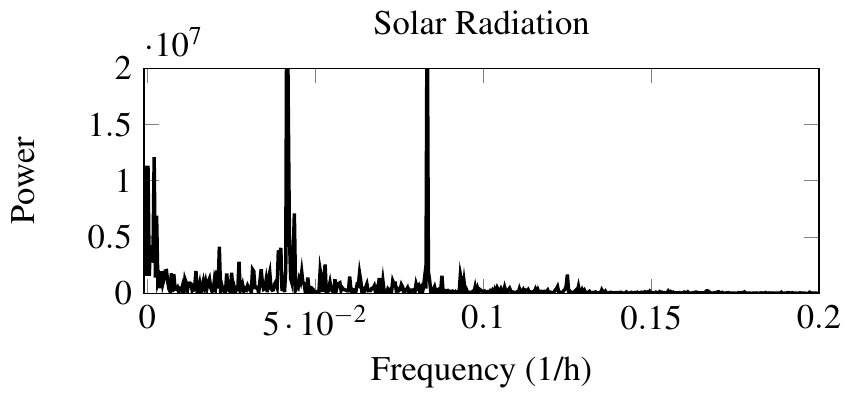}
\caption{Positive Spectrum of Input and Disturbance Signals.}
\label{fig:Spectrum}
\end{figure}

For either the simplified model of (\ref{eq:LinearRegressorMI}) or the more detailed model of (\ref{eq:NonlinearRegressorMI}), note that any transfer function (from an input/disturbance to the zone temperature) requires the identification of at most $2(\magn{\mathcal{N}_i}+2)$ parameters. Thus, in order for an experiment to be informative enough, it is sufficient that the considered inputs/disturbances are persistently exciting of order $2(\magn{\mathcal{N}_i}+2)$ or higher. In the case of $\magn{\mathcal{N}_i}=1$ (considered in this experiment), it suffices that any input/disturbance signal is persistently exciting of order $6$ (or $5$ in case one of the frequencies is at zero, as it is the case for all signals in Figure~\ref{fig:Spectrum}). As demonstrated in Figure~\ref{fig:Spectrum}, this condition is satisfied by the input/disturbance signals, since for all of them the \emph{positive} spectrum is non-zero in at least $3$ distinct frequencies.

\subsection{Identification error comparison}		\label{sec:IdentificationErrorComparison}

In Figure~\ref{fig:IdentificationComparison}, we demonstrate the resulting identification error under the LRM of (\ref{eq:LinearRegressorMI}) and the NRM of (\ref{eq:NonlinearRegressorMI}) for $\magn{\mathcal{N}_i}=1$. Note that the NRM achieves a smaller identification error of about $10\%$. This reduced error is observed later in time (due to the larger training time required by the larger number of terms used in the nonlinear regression vector). The data used for training correspond to data collected between October and April from the simulated building, however the same data have been reused several times for better fitting. Thus, the time axis in Figure~\ref{fig:IdentificationComparison} does \emph{not} correspond to the actual simulation time, but it corresponds to the accumulated time index of the reused data.

\begin{figure}[h!]
\centering
\includegraphics{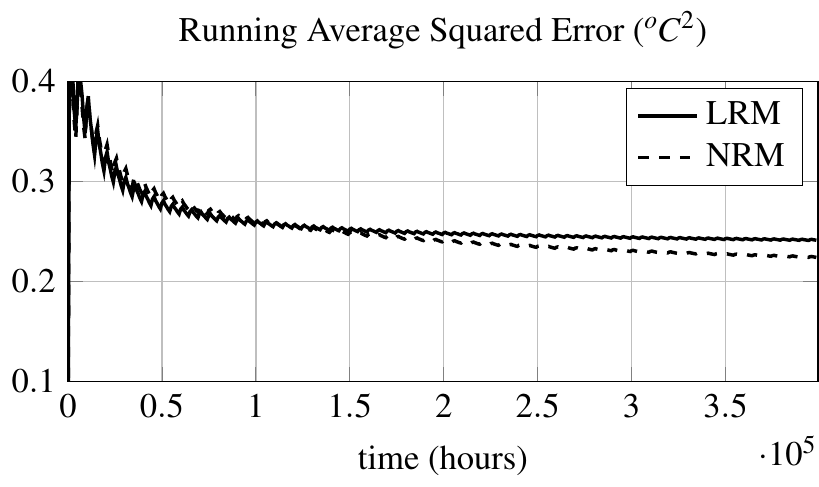}
\caption{Identification performance for the zone temperature.}
\label{fig:IdentificationComparison}
\end{figure}

\subsection{Prediction performance comparison}	\label{sec:PredictionPerformanceComparison}

Small identification error may not necessarily guarantee a good prediction performance, especially when predictions are required for several hours ahead. From the one hand, training for system identification is based upon the one-step ahead innovations, which does not necessarily guarantee a similar accurate prediction for several hours ahead. On the other hand, training might have been performed under operating conditions that are (possibly) different from the operating conditions under which predictions are requested (e.g., weather conditions may dramatically change within a few days). This may potentially result in a significant degradation of the prediction performance when the prediction has not taken into account the effect of the nonlinearities present in the dynamics. The effect of these nonlinearities may not be known a-priori and a thorough investigation is necessary to understand their potential economic impact.

To this end, we designed a standard MPC for the RH system of the main living area of the residential building. The goal is to evaluate the performance of the previously designed (nonlinear) regression vectors for the zone temperature prediction compared to standard (black-box) linear regression vectors. The structure of the MPC is rather simple and addresses the following optimization problem.

\begin{subequations}	\label{eq:MPCformulation}
\begin{align}
\min &&& \alpha\sum_{k=0}^{N_{\rm hor}}\Big\{\hat{p}_i(k) \left(\hat{T}_{r,i}(k)-T_{\rm set}(k)\right)^2/N_{\rm hor}\Big\} + \nonumber \\ 
&&& \sum_{k=0}^{N_{\rm hor}-1}\Big\{ \beta T_{\rm sam} \left(T_{w,i}^{+}(k)-\hat{T}_{w,i}^{-}(k)\right)\Big) + \gamma T_{\rm sam} \dot{V}_{w,i}(k)  \Big\} \label{eq:MPCObjectiveFunction} \\ 
\mbox{s.t.} &&& \hat{T}_{r,i}(k) \approx \varphi_{r,i}(k)\tr \theta_{r,i} \label{eq:MPCTemperaturePrediction} \\ 
&&& \hat{T}_{w,i}^{-}(k) \approx \hat{T}_{w,i}(k) \approx \varphi_{w,i}(k)\cdot\theta_{w,i} \\
\mbox{var.} &&& T_{w,i}^{+}(k)\in\{40^{o}C,45^{o}C\}, \\
&&& \dot{V}_{w,i}(k)\in\{0,\dot{V}_{w,{\rm max}}\}, \\
&&& k=0,1,2,...,N_{\rm hor}-1, \nonumber
\end{align}
\end{subequations}


Note that the first part of the objective function (\ref{eq:MPCObjectiveFunction}) corresponds to a \emph{comfort} cost. It measures the average squared difference of the zone temperature from the desired (or set) temperature entered by the user at time $k$. The set temperature was set equal to $21^oC$ throughout the optimization horizon. The variable $\hat{p}_i(k)\in\{0,1\}$ holds our estimates on whether people are present in zone $i$ at time instance $k$. 

The second part of the objective function (\ref{eq:MPCObjectiveFunction}) corresponds to the \emph{heating cost}, while the third part corresponds to the \emph{pump-electricity cost}. The nonnegative parameters $\beta$, $\gamma$ were previously identified for the heating system of the simulated building and take values: $\beta= 0.3333 kW/^oCh$, $\gamma=0.5278\cdot{10}^{3} kWsec/h m^3 $. The non-negative constant $\alpha$ is introduced to allow for adjusting the importance of the comfort cost compared to the energy cost. A large value has been assigned equal to $10^6$ to enforce high comfort.

The $\hat{p}_i(k)$, $k=1,2,...,$, as well as the outdoor temperature, $T_{\rm out}$, and the solar gain, $\dot{Q}_{\rm ext}(k)$, are assumed given (i.e., predicted with perfect accuracy). This assumption is essential in order to evaluate precisely the impact of our temperature predictions (\ref{eq:MPCTemperaturePrediction}) in the performance of the derived optimal controller. 

The sampling period was set to $T_{\rm sam} = \nicefrac{1}{12}h$, the optimization period was set to $T_{\rm opt}=1h$, and the optimization horizon was set to $T_{\rm hor}=5h$. This implies that $N_{\rm hor}=5 \cdot 12 = 60$. Furthermore, the control variables are the inlet water temperature which assumes two values ($40^oC$ and $45^oC$) and the water flow which assumes only two values, the minimum $0 kg/sec$ and the maximum $\dot{V}_{w,\max}=0.0787kg/sec$. 

For the prediction model of (\ref{eq:MPCTemperaturePrediction}), we used either the LRM of (\ref{eq:LinearRegressorMI}) or the NRM of (\ref{eq:NonlinearRegressorMI}), both trained offline using data collected during normal operation of a hysteresis controller (as demonstrated in detail in Section~\ref{sec:DataGeneration}). For the prediction model of the outlet water temperature used in the computation of the cost function, we used the prediction model derived in (\ref{eq:RHpredictorFI}). As in the case of Section~\ref{sec:IdentificationErrorComparison} where the LRM and the NRM were compared, we evaluated the value of the cost function under these two alternatives. The performances of the two models with respect to the comfort cost (which corresponds to the 1st term of the objective (\ref{eq:MPCObjectiveFunction})) are presented in Figure~\ref{fig:ComfortErrorComparison}. The performances of the two models with respect to the energy spent (i.e., heating and pump electricity cost), which correspond to the 2nd and 3rd term of the objective (\ref{eq:MPCObjectiveFunction}), are presented in Figure~\ref{fig:CostComparison}. Note that the NRM achieves lower comfort cost using less amount of energy which is an indication of small prediction errors.

\begin{figure}[th!]
\centering
\includegraphics{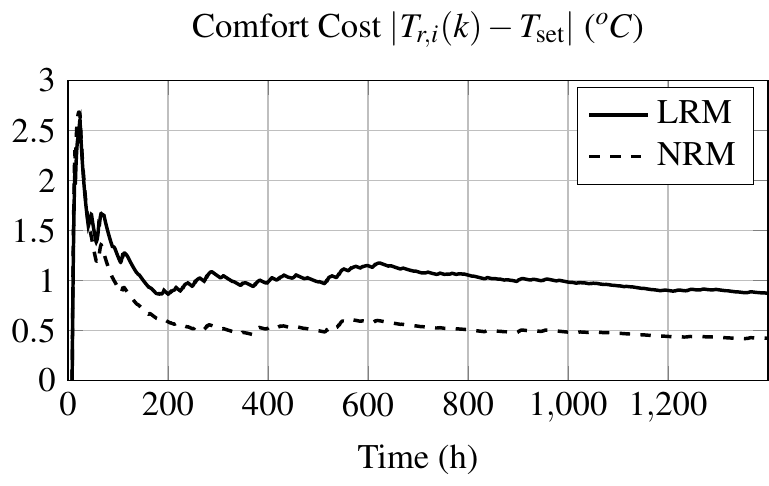}
\caption{Running-average comfort cost under LRM and NRM.}
\label{fig:ComfortErrorComparison}
\end{figure}
\begin{figure}[th!]
\centering
\includegraphics{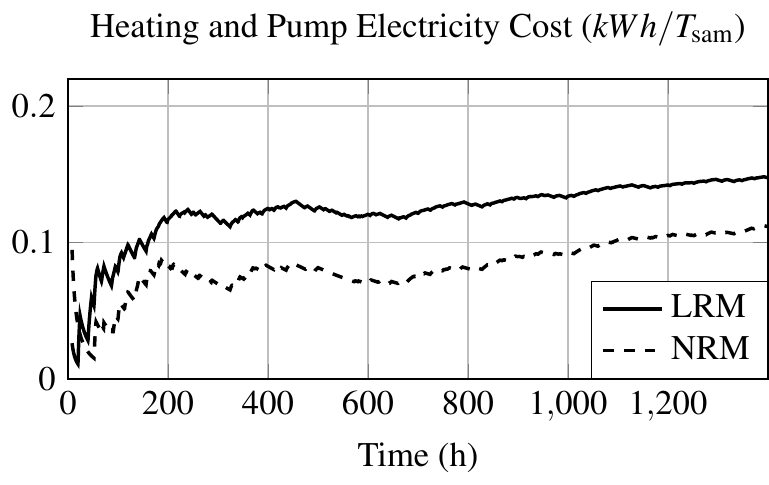}
\caption{Running-average heating and pump electricity cost under LRM and NRM.}
\label{fig:CostComparison}
\end{figure}

\subsection{Discussion}

Note that the NRM achieves an improvement to both prediction accuracy and cost performance. The fact that the LRM may not provide a similarly good performance as the NRM could be attributed to the fact that the operating conditions for testing the MPC controller might not be identical with the training conditions. Note, for example, that the training phase was performed under the heating-curve pattern of (\ref{eq:HeatingCurve}), while the MPC employs only two alternative water inlet temperatures. This variation in the inlet water temperature may have exhibited different patterns in the nonlinear parts of the dynamics, thus affecting the prediction performance. 

The simulation study presented in this paper considered a regression vector generated with $\magn{\mathcal{N}_i}=1$, which corresponds to the case of a single neighboring zone (e.g., the outdoor temperature). The proposed regression models are also applicable for higher-order models (i.e., when more than one neighboring zones are present), however it might be necessary that some of the input signals are designed to include larger number of frequencies. Note, however, that the input signals that can be designed (e.g., medium flow rates and inlet temperatures) appear as products with other disturbance variables in the derived regression models (e.g., equation (\ref{eq:NonlinearRegressorMI})). Thus, persistence of excitation for such products of non-correlated signals is an easier task. On the other hand, some of the disturbance signals (e.g., the solar radiation or the outdoor temperature) cannot be designed. Of course, the prediction model will operate at conditions that are similar to which it was identified. However, the impact of such disturbance signals when they are not informative enough needs to be further investigated.


\section{Conclusions} \label{sec:Conclusions}

We derived analytically regression models admitting nonlinear regressors specifically tailored for system identification and prediction for thermal dynamics in buildings. The proposed models were compared with standard (black-box) linear regression models derived through the OE structure, and an improvement was observed with respect to both the resulting prediction performance as well as the energy cost. The greatest advantage of the proposed identification scheme relies on the fact that it provides a richer and more accurate representation of the underlying physical phenomena, contrary to standard black-box identification schemes. 
\bibliographystyle{asmems4}

\bibliography{mpcEnergy_SystemID_Bibliography.bib}

\appendix

\section{Proof of Property~\ref{Pr:Property2}} \label{Ap:Property2}

We can write:
\begin{eqnarray*}
\lefteqn{[\QQ_{s,ij}\QQ_{wc,i}(\dot{V}_{w,i}(k))]\left\{x(k)\right\} } \cr & = & \QQ_{s,ij}\left\{\QQ_{wc,i}(\dot{V}_{w,i}(k))\left\{x(k)\right\}\right\} \cr 
& = & \QQ_{s,ij}\left\{ x(k) - (1 + \epsilon a_{wc,i}(\dot{V}_{w,i}(k)))x(k-1) \right\} \cr 
& = & x(k) - (1+\epsilon a_{wc,i}(\dot{V}_{w,i}(k)))x(k-1) - \cr && (1+\epsilon a_{s,ij})x(k-1) + (1+\epsilon a_{s,ij}) x(k-2) + \cr && (1+\epsilon a_{s,ij}) \epsilon a_{wc,i}(\dot{V}_{w,i}(k-1))x(k-2). 
\end{eqnarray*}
Similarly, we have:
\begin{eqnarray*}
\lefteqn{[\QQ_{wc,i}(\dot{V}_{w,i}(k))\QQ_{s,ij}]\left\{x(k)\right\} } \cr 
& = & x(k) - (1+\epsilon a_{wc,i}(\dot{V}_{w,i}(k)))x(k-1) - \cr && (1+\epsilon a_{s,ij})x(k-1) + (1+\epsilon a_{s,ij}) x(k-2) + \cr && (1+\epsilon a_{s,ij}) \epsilon a_{wc,i}(\dot{V}_{w,i}(k))x(k-2). 
\end{eqnarray*}
Thus, we have:
\begin{eqnarray*}
\lefteqn{[\QQ_{s,ij}\QQ_{wc,i}(\dot{V}_{w,i}(k))]\left\{x(k)\right\} = } \cr  && [\QQ_{wc,i}(\dot{V}_{w,i}(k))\QQ_{s,ij}]\left\{x(k)\right\} - \cr && (1+\epsilon a_{s,ij})\epsilon\left( a_{wc,i}(\dot{V}_{w,i}(k)) - a_{wc,i}(\dot{V}_{w,i}(k-1)) \right) x(k-2).
\end{eqnarray*}
Finally, we may write:
\begin{eqnarray}
\lefteqn{[\QQ_{s,ij}\QQ_{wc,i}(\dot{V}_{w,i}(k))] = } \cr  && [\QQ_{wc,i}(\dot{V}_{w,i}(k))\QQ_{s,ij}] - \cr && (1+\epsilon a_{s,ij})\epsilon \left[(1-\qq^{-1})a_{wc,i}(\dot{V}_{w,i}(k)) \right] \qq^{-2}.
\end{eqnarray}

\section{Proof of Proposition~\ref{Pr:NRM:ZoneDynamicsFI}} \label{Ap:NRM:ZoneDynamicsFI}

A Taylor-series expansion of the continuous-time zone dynamics leads to the following approximation:
\begin{eqnarray}	\label{eq:NRM:ZoneDynamicsFI1}
\lefteqn{T_{r,i}(k) \approx }\cr && \PP_{r,i}(\dot{V}_{a,i}(k-1))\{T_{r,i}(k)\} +  \epsilon\sum_{j\in\mathcal{N}_i}a_{rs,ij}^{+}T_{s,ij}(k-1) + \cr && \epsilon a_{rw,i}T_{w,i}(k-1) + \epsilon a_{ra,i}(\dot{V}_{a,i}(k-1))T_{a,i}^{+}(k-1) + \cr && \epsilon a_{{\rm ext},i}\dot{Q}_{{\rm ext},i}(k-1),
\end{eqnarray}
plus higher-order terms of $\epsilon$. Furthermore, note that the finite response of the separator-dynamics can be approximated by:
\begin{eqnarray}	\label{eq:NRM:ZoneDynamicsFI1_approximation}
T_{s,ij}(k) \approx \QQ_{s,ij}^{-1}\Big\{ \epsilon a_{s,ij}^{+}T_{r,i}(k-1) + \epsilon a_{s,ij}^{-}T_{r,j}(k-1)\Big\}
\end{eqnarray}
plus higher-order terms of $\epsilon$.
Given that the family of operators $\QQ_{s,ij}$, $j\in\mathcal{N}_i$ are pairwise commutative by Property~\ref{Pr:Property1}, if we replace $T_{s,ij}$ into (\ref{eq:NRM:ZoneDynamicsFI1}) and we multiply both sides of the above expression by the composite operator $\QQ_{s,ij_1}\cdots\mathcal{Q}_{s,ij_{\magn{\mathcal{N}_i}}}$, briefly denoted by $\prod_{j\in\mathcal{N}_i}\QQ_{s,ij}\{\cdot\}$, we have:
\begin{eqnarray*}	\label{eq:NRM:ZoneDynamicsFI2}
\lefteqn{T_{r,i}(k) \approx} \cr && \Big[1-\prod_{j\in\mathcal{N}_i}\QQ_{s,ij}\Big]\Big\{T_{r,i}(k)\Big\} + \cr && \Big[\prod_{j\in\mathcal{N}_i}\QQ_{s,ij}\Big]\left\{\PP_{r,i}(\dot{V}_{a,i}(k-1))\{T_{r,i}(k)\}\right\} + \cr && \epsilon^2\sum_{j\in\mathcal{N}_i}a_{rs,ij}^{+}\Big[\prod_{\ell\neq{j}}\QQ_{s,i\ell}\Big]\left\{a_{s,ij}^{+}T_{r,i}(k-2)\right\}+ \cr && 
\epsilon^2\sum_{j\in\mathcal{N}_i}a_{rs,ij}^{+}\Big[\prod_{\ell\neq{j}}\QQ_{s,i\ell}\Big]\left\{a_{s,ij}^{-}T_{r,j}(k-2)\right\} + \cr && \epsilon a_{rw,i}\Big[\prod_{j\in\mathcal{N}_i}\QQ_{s,ij}\Big]\Big\{T_{w,i}(k-1)\Big\} + \cr && \epsilon \Big[\prod_{j\in\mathcal{N}_i}\QQ_{s,ij}\Big]\left\{ a_{ra,i}(\dot{V}_{a,i}(k-1))T_{a,i}^{+}(k-1)\right\} + \cr && \epsilon a_{{\rm ext},i}\Big[\prod_{j\in\mathcal{N}_i}\QQ_{s,ij}\Big]\left\{\dot{Q}_{{\rm ext},i}(k-1)\right\},
\end{eqnarray*}
plus higher-order terms of $\epsilon$. 
The conclusion follows directly by using Property~\ref{Pr:Property3} and expanding the terms of the above expression. 

\section{Proof of Proposition~\ref{Pr:NRM:ZoneDynamicsMI}} \label{Ap:NRM:ZoneDynamicsMI}

Under the MI structure, the finite-response approximation of (\ref{eq:NRM:ZoneDynamicsFI1}) continues to hold for sufficiently small step-size sampling interval $\epsilon$. Note that the separator dynamics can still be approximated according to Equation~(\ref{eq:NRM:ZoneDynamicsFI1_approximation}).
Furthermore, a Taylor-series expansion of the water dynamics leads to the following finite-response approximation:
\begin{eqnarray*}
\lefteqn{T_{w,i}(k) \approx \QQ_{wc,i}^{-1}(\dot{V}_{w,i}(k-1)) }\cr && 
\Big\{\epsilon a_{ww,i}(\dot{V}_{w,i}(k-1))T_{w,i}^{+}(k-1) +  \epsilon a_{w,i}T_{r,i}(k-1)\Big\}
\end{eqnarray*}
plus higher-order terms of $\epsilon$. If we replace $T_{s,ij}$ and $T_{w,i}$  into (\ref{eq:NRM:ZoneDynamicsFI1}) and we apply to both sides of the resulting expression the operator $\QQ_{wc,i}(\dot{V}_{w,i}(k-2))$, we get
\begin{eqnarray*}
\lefteqn{\QQ_{wc,i}(\dot{V}_{w,i}(k-2))\Big\{T_{r,i}(k)\Big\} \approx } \cr 
&& \QQ_{wc,i}(\dot{V}_{w,i}(k-2))\left\{\PP_{r,i}(\dot{V}_{a,i}(k-1))\{T_{r,i}(k)\}\right\} + \cr 
&& \epsilon^2\sum_{j\in\mathcal{N}_i}a_{rs,ij}^{+}\QQ_{wc,i}(\dot{V}_{w,i}(k-2))\QQ_{s,ij}^{-1}\left\{a_{s,ij}^{+}T_{r,i}(k-2)\right\}+ \cr 
&& \epsilon^2\sum_{j\in\mathcal{N}_i}a_{rs,ij}^{+}\QQ_{wc,i}(\dot{V}_{w,i}(k-2))\QQ_{s,ij}^{-1}\left\{a_{s,ij}^{-}T_{r,j}(k-2)\right\} + \cr 
&& \epsilon^2 a_{rw,i}a_{ww,i}(\dot{V}_{w,i}(k-2))T_{w,i}^{+}(k-2) + \cr
&& \epsilon^2 a_{rw,i}a_{w,i}T_{r,i}(k-2) + \cr
&& \epsilon \QQ_{wc,i}(\dot{V}_{w,i}(k-2))\left\{ a_{ra,i}(\dot{V}_{a,i}(k-1))T_{a,i}^{+}(k-1)\right\} + \cr 
&& \epsilon a_{{\rm ext},i} \QQ_{wc,i}(\dot{V}_{w,i}(k-2))\left\{\dot{Q}_{{\rm ext},i}(k-1)\right\}
\end{eqnarray*}
plus higher-order terms of $\epsilon$. Applying to both sides of the above expression the composite operator $\QQ_{s,ij_1}\cdots\mathcal{Q}_{s,ij_{\magn{\mathcal{N}_i}}}$, briefly denoted by $\prod_{j\in\mathcal{N}_i}\QQ_{s,ij}\{\cdot\}$, and making use of Property~\ref{Pr:Property1}, we have:
\begin{eqnarray*}
\lefteqn{\prod_{j\in\mathcal{N}_i}\QQ_{s,ij}\QQ_{wc,i}(\dot{V}_{w,i}(k-2))\Big\{T_{r,i}(k)\Big\} \approx } \cr 
&& \prod_{j\in\mathcal{N}_i}\QQ_{s,ij}\QQ_{wc,i}(\dot{V}_{w,i}(k-2))\left\{\PP_{r,i}(\dot{V}_{a,i}(k-1))\{T_{r,i}(k)\}\right\} + \cr 
&& \epsilon^2\sum_{j\in\mathcal{N}_i}a_{rs,ij}^{+}\prod_{\ell\in\mathcal{N}_i\backslash{j}}\QQ_{s,i\ell}[\QQ_{s,ij}\QQ_{wc,i}(\dot{V}_{w,i}(k-2))]\QQ_{s,ij}^{-1} \cr 
&& \left\{a_{s,ij}^{+}T_{r,i}(k-2)\right\} + \cr 
&& \epsilon^2\sum_{j\in\mathcal{N}_i}a_{rs,ij}^{+}\prod_{\ell\in\mathcal{N}_i\backslash{j}}\QQ_{s,i\ell}[\QQ_{s,ij}\QQ_{wc,i}(\dot{V}_{w,i}(k-2))]\QQ_{s,ij}^{-1} \cr 
&& \left\{a_{s,ij}^{-}T_{r,j}(k-2)\right\} + \cr 
&& \epsilon^2 a_{rw,i}\prod_{j\in\mathcal{N}_i}\QQ_{s,ij}\left\{a_{ww,i}(\dot{V}_{w,i}(k-2))T_{w,i}^{+}(k-2)\right\}  + \cr
&& \epsilon^2 a_{rw,i}a_{w,i}\prod_{j\in\mathcal{N}_i}\QQ_{s,ij}\left\{T_{r,i}(k-2)\right\} + \cr
&& \epsilon \prod_{j\in\mathcal{N}_i}\QQ_{s,ij} \QQ_{wc,i}(\dot{V}_{w,i}(k-2)) \cr 
&& \left\{ a_{ra,i}(\dot{V}_{a,i}(k-1))T_{a,i}^{+}(k-1)\right\} + \cr 
&& \epsilon a_{{\rm ext},i} \prod_{j\in\mathcal{N}_i}\QQ_{s,ij} \QQ_{wc,i}(\dot{V}_{w,i}(k-2))\left\{\dot{Q}_{{\rm ext},i}(k-1)\right\} 
\end{eqnarray*}
plus higher-order terms of $\epsilon$. According to Property~\ref{Pr:Property2}, we may perform the substitution:
\begin{eqnarray*}
\lefteqn{[\QQ_{s,ij}\QQ_{wc,i}(\dot{V}_{w,i}(k-2))] =}\cr && [\QQ_{wc,i}(\dot{V}_{w,i}(k-2))\QQ_{s,ij}] - \cr &&
(1+\epsilon a_{s,ij})\epsilon \left[(1-\qq^{-1})a_{wc,i}(\dot{V}_{w,i}(k-2))\right]\qq^{-2}.
\end{eqnarray*}
Thus, we may write:
\begin{eqnarray}	\label{eq:ApproximationMI:Equation1}
\lefteqn{\prod_{j\in\mathcal{N}_i}\QQ_{s,ij}\QQ_{wc,i}(\dot{V}_{w,i}(k-2))\QQ_{r,i}(\dot{V}_{a,i}(k-1))\Big\{T_{r,i}(k)\Big\} \approx} \cr 
&& \epsilon^2\sum_{j\in\mathcal{N}_i}a_{rs,ij}^{+}\prod_{\ell\in\mathcal{N}_i\backslash{j}}\QQ_{s,i\ell}\QQ_{wc,i}(\dot{V}_{w,i}(k-2)) \cr 
&& \left\{a_{s,ij}^{+}T_{r,i}(k-2)\right\}+ \cr 
&& \epsilon^2\sum_{j\in\mathcal{N}_i}a_{rs,ij}^{+}\prod_{\ell\in\mathcal{N}_i\backslash{j}}\QQ_{s,i\ell}\QQ_{wc,i}(\dot{V}_{w,i}(k-2)) \cr 
&& \left\{a_{s,ij}^{-}T_{r,j}(k-2)\right\} + \cr 
&& \epsilon^2 a_{rw,i}\prod_{j\in\mathcal{N}_i}\QQ_{s,ij}\left\{a_{ww,i}(\dot{V}_{w,i}(k-2))T_{w,i}^{+}(k-2)\right\} +  \cr
&& \epsilon^2 a_{rw,i} a_{w,i} \prod_{j\in\mathcal{N}_i} \QQ_{s,ij}\left\{T_{r,i}(k-2)\right\} + \cr
&& \epsilon \prod_{j\in\mathcal{N}_i}\QQ_{s,ij} \QQ_{wc,i}(\dot{V}_{w,i}(k-2)) \cr 
&& \left\{ a_{ra,i}(\dot{V}_{a,i}(k-1))T_{a,i}^{+}(k-1)\right\} + \cr 
&& \epsilon a_{{\rm ext},i} \prod_{j\in\mathcal{N}_i}\QQ_{s,ij} \QQ_{wc,i}(\dot{V}_{w,i}(k-2))\left\{\dot{Q}_{{\rm ext},i}(k-1)\right\},
\end{eqnarray}
plus higher-order terms of $\epsilon$. In the following, we will approximate the terms of the above expression, by neglecting terms of order of $\varepsilon^3$ or higher. 

Note that
\begin{eqnarray}	\label{eq:ApproximationMI}
\lefteqn{\QQ_{wc,i}(\dot{V}_{w,i}(k-2))\QQ_{r,i}(\dot{V}_{a,i}(k-1))\{T_{r,i}(k)\} = } \cr 
&& T_{r,i}(k) - 2T_{r,i}(k-1) + T_{r,i}(k-2) - \cr && 
\epsilon a_{r,i}(\dot{V}_{a,i}(k-1))T_{r,i}(k-1) - \cr &&
\epsilon a_{wc,i}(\dot{V}_{w,i}(k-2)) T_{r,i}(k-1) + \cr && 
\epsilon a_{wc,i}(\dot{V}_{w,i}(k-2))T_{r,i}(k-2) + \cr &&
\epsilon a_{r,i}(\dot{V}_{a,i}(k-2)) T_{r,i}(k-2) + \cr &&
\epsilon^2 a_{wc,i}(\dot{V}_{w,i}(k-2))a_{r,i}(\dot{V}_{a,i}(k-2))T_{r,i}(k-2).
\end{eqnarray}
The following also hold:
\begin{align*}
&\QQ_{wc,i}(\dot{V}_{w,i}(k-2))\left\{a_{s,ij}^{+}T_{r,i}(k-2)\right\}  \cr 
= & a_{s,ij}^{+} T_{r,i}(k-2) - (1+\epsilon a_{wc,i}(\dot{V}_{w,i}(k-2)))a_{s,ij}^{+}T_{r,i}(k-3) \cr 
\approx & a_{s,ij}^{+} T_{r,i}(k-2) - a_{s,ij}^{+}T_{r,i}(k-3) \cr\cr
%
%
& \QQ_{wc,i}(\dot{V}_{w,i}(k-2))\left\{a_{ra,i}(\dot{V}_{a,i}(k-1)T_{a,i}^+(k-1)\right\} \cr 
= & a_{ra,i}(\dot{V}_{a,i}(k-1))T_{a,i}^{+}(k-1) - \cr 
& (1+\epsilon a_{wc,i}(\dot{V}_{w,i}(k-2)))a_{ra,i}(\dot{V}_{a,i}(k-2))T_{a,i}^{+}(k-2) \cr\cr
%
& \QQ_{wc,i}(\dot{V}_{w,i}(k-2))\left\{\dot{Q}_{{\rm ext},i}(k-1)\right\} \cr 
= & \dot{Q}_{{\rm ext},i}(k-1) -  
(1 + \epsilon a_{wc,i}(\dot{V}_{w,i}(k-2)))\dot{Q}_{{\rm ext},i}(k-2) 
\end{align*}
where the first expression ignores terms of $\epsilon$ since, in (\ref{eq:ApproximationMI:Equation1}), it multiplies an expression of order of $\varepsilon^2$. Using the above approximations and Property~\ref{Pr:Property3}, the terms of (\ref{eq:ApproximationMI:Equation1}) can be approximated as follows:
\begin{align*}
&\prod_{j\in\mathcal{N}_i}\QQ_{s,ij}\QQ_{wc,i}(\dot{V}_{w,i}(k-2))\QQ_{r,i}(\dot{V}_{a,i}(k-1))\Big\{T_{r,i}(k)\Big\} \cr = & T_{r,i}(k) + \sum_{m=1}^{\magn{\mathcal{N}_i}+2}\alpha_{1,m}^{(1)} T_{r,i}(k-m) + \cr 
& \sum_{m=1}^{\magn{\mathcal{N}_{i}}+2}\alpha_{2,m}^{(1)} \dot{V}_{a,i}(k-m)T_{r,i}(k-m) + \cr
& \sum_{m=1}^{\magn{\mathcal{N}_{i}}+1}\alpha_{3,m}^{(1)} \dot{V}_{w,i}(k-1-m)T_{r,i}(k-m) + \cr 
& \sum_{m=2}^{\magn{\mathcal{N}_{i}}+2}\alpha_{4,m}^{(1)} \dot{V}_{w,i}(k-m)T_{r,i}(k-m) + \cr
& \sum_{m=2}^{\magn{\mathcal{N}_{i}}+2}\alpha_{5,m}^{(1)} \dot{V}_{w,i}(k-m)\dot{V}_{a,i}(k-m)T_{r,i}(k-m),\cr\cr
& \prod_{\ell\in\mathcal{N}_i\backslash{j}}\QQ_{s,i\ell}\QQ_{wc,i}(\dot{V}_{w,i}(k-2))\left\{a_{s,ij}^{+}T_{r,i}(k-2)\right\} \cr 
\approx & \sum_{m=2}^{\magn{\mathcal{N}_i}+2}\alpha_{1,m}^{(2)}T_{r,i}(k-m), 
\cr\cr
%
%
& \prod_{j\in\mathcal{N}_i}\QQ_{s,ij}\left\{a_{ww,i}(\dot{V}_{w,i}(k-2))T_{w,i}^{+}(k-2)\right\} \cr 
= & \sum_{m=2}^{\magn{\mathcal{N}_i}+2}\alpha_{1,m}^{(4)}\dot{V}_{w,i}(k-m)T_{w,i}^{+}(k-m), \cr\cr
& \prod_{j\in\mathcal{N}_i}\QQ_{s,ij}\left\{T_{r,i}^{+}(k-2)\right\} \cr 
= & \sum_{m=2}^{\magn{\mathcal{N}_i}+2}\alpha_{1,m}^{(5)}T_{r,i}^{+}(k-m), \cr\cr
& \prod_{j\in\mathcal{N}_i}\QQ_{s,ij} \QQ_{wc,i}(\dot{V}_{w,i}(k-2))\left\{ a_{ra,i}(\dot{V}_{a,i}(k-1))T_{a,i}^{+}(k-1)\right\} \cr
= & \sum_{m=1}^{\magn{\mathcal{N}_i}+2}\alpha_{1,m}^{(5)}\dot{V}_{a,i}(k-m)T_{a,i}^{+}(k-m) + \cr 
& \sum_{m=2}^{\magn{\mathcal{N}_i}+2}\alpha_{2,m}^{(5)}\dot{V}_{w,i}(k-m)\dot{V}_{a,i}(k-m)T_{a,i}^{+}(k-m), \cr\cr
&\prod_{j\in\mathcal{N}_i}\QQ_{s,ij} \QQ_{wc,i}(\dot{V}_{w,i}(k-2))\left\{\dot{Q}_{{\rm ext},i}(k-1)\right\} \cr
= & \sum_{m=1}^{\magn{\mathcal{N}_i}+2}\alpha_{1,m}^{(6)}\dot{Q}_{{\rm ext},i}(k-m) + \cr 
& \sum_{m=2}^{\magn{\mathcal{N}_i}+2}\alpha_{2,m}^{(6)}\dot{V}_{w,i}(k-m)\dot{Q}_{{\rm ext},i}(k-m).
\end{align*}
for some constant parameters $\alpha_{\times,m}^{(\times)}\in\mathbb{R}$.

It is straightforward to check that using the above approximations, Equation~(\ref{eq:ApproximationMI:Equation1}) can be written as a linear regression with a nonlinear regression vector of the form depicted in Equation~(\ref{eq:NonlinearRegressorMI}).



\end{document}